\documentclass[3p,authoryear]{elsarticle}

\usepackage[utf8]{inputenc}
\usepackage{rotfloat}
\usepackage{booktabs}
\usepackage{amsmath}
\usepackage{amsthm}
\usepackage{amssymb}
\usepackage{graphicx}

\makeatletter

\providecommand{\tabularnewline}{\\}

\theoremstyle{plain}
\newtheorem{thm}{\protect\theoremname}

\journal{arXiv}

\@ifundefined{showcaptionsetup}{}{%
 \PassOptionsToPackage{caption=false}{subfig}}
\usepackage{subfig}
\makeatother

\providecommand{\theoremname}{Theorem}

\begin{document}

\begin{frontmatter}{}

\title{Consistent and Efficient Pricing of SPX and VIX Options under Multiscale
Stochastic Volatility}

\author[labela]{Jaegi Jeon}

\author[labelb]{Geonwoo Kim}

\author[labelc]{Jeonggyu Huh\corref{cor1}}

\address[labela]{Department of Mathematics, Yonsei University, Seoul 03722, Republic
of Korea}

\address[labelb]{School of Liberal Arts, Seoul National University of Science and
Technology, Seoul 01811, Republic of Korea}

\address[labelc]{School of Computational Sciences, Korea Institute for Advanced Study,
Seoul 02455, Republic of Korea}

\cortext[cor1]{corresponding author. Tel: 82-2-958-3750\\
 E-mail address: aifina2018@kias.re.kr}
\begin{abstract}
This study provides a consistent and efficient pricing method for
both Standard \& Poor's 500 Index (SPX) options and the Chicago Board
Options Exchange's Volatility Index (VIX) options under a multiscale
stochastic volatility model. To capture the multiscale volatility
of the financial market, our model adds a fast scale factor to the
well-known Heston volatility and we derive approximate analytic pricing
formulas for the options under the model. The analytic tractability
can greatly improve the efficiency of calibration compared to fitting
procedures with the finite difference method or Monte Carlo simulation.
Our experiment using options data from 2016 to 2018 shows that the
model reduces the errors on the training sets of the SPX and VIX options
by 9.9\% and 13.2\%, respectively, and decreases the errors on the
test sets of the SPX and VIX options by 13.0\% and 16.5\%, respectively,
compared to the single-scale model of Heston. The error reduction
is possible because the additional factor reflects short-term impacts
on the market, which is difficult to achieve with only one factor.
It highlights the necessity of modeling multiscale volatility.
\end{abstract}

\end{frontmatter}{}

\section{Introduction}

Volatility in the financial market is one of the most important measures
for investors. The Volatility Index (VIX) of the Chicago Board Options
Exchange (CBOE) was introduced as a measure to capture the volatility
of the Standard \& Poor's 500 Index (SPX) in 1993. The VIX was revised
to be independent of any model in 2003. After the revision, it is
calculated using SPX option prices and the SPX over the next 30 calendar-day
period. Since the introduction of the VIX, the index has become a
standard gauge of market fear, because the VIX and the SPX are negatively
correlated. In addition, VIX options also have received a lot of attention
as the popularity of the VIX has grown, and this interest in the derivatives
has led to the development of an appropriate pricing method for them.

It is noteworthy that many researchers have invented pricing methods
for VIX options under the models originally proposed to price SPX
options. This may be because an arbitrage relationship exists between
the SPX option market and the VIX option market. In the early research
period for VIX options, there were the studies to derive pricing formulas
for VIX options under several one-factor affine jump diffusion (AJD)
models; \citet{lin2009vix}, \citet{lian2013pricing}, \citet{lin2010consistent},
\citet{cheng2012remark} and \citet{cont2013consistent}. As it is
widely known, the one-factor AJD model has a special structure from
a mathematical standpoint, so they often give a pricing formula for
a derivative for SPX options as well as VIX options \citep[refer][]{duffie2000transform}.
It is highly desirable in the sense that they make it possible to
evaluate the SPX options and the VIX options without arbitrage. However,
it is recognized that the one-factor AJD models have a variety of
fundamental limitations. Above all, the models cannot accommodate
the multiscale property of volatility, which highlights the necessity
of multi-factor affine diffusion models. Jumps are usually excluded
in the models, but it is frequently stated that the models are superior
to the one-factor AJD models in literature such as \citet{kaeck2012volatility}.

In this light, \citet{gatheral2008consistent} proposed the double-mean-reverting
(DMR) Heston model to capture the multiscales of volatility. But it
does not provide any pricing formula and depends on the Monte-Carlo
simulation totally for a pricing and a calibration. Although the simulation
method has been greatly improved in \citet{bayer2013fast}, the model
still fails to suggest a viable solution to a calibration using a
large number of option prices. On the other hand, \citet{fouque2018heston}
proposed the Heston stochastic vol-of-vol model (Heston SVV model)
with fast and slow time scales to reflect skews in both the SPX data
and VIX data, and derived approximation solutions for the prices of
options on the SPX and VIX, respectively. Moreover, it showed that
the Heston SVV model calibrated jointly to real market data for the
SPX and VIX fits to the data very well. Their solutions, however,
cannot be also computed within a short period of time. Interestingly,
\citet{huh2018scaled} obtained an approximate quasi-closed formula
for VIX options under the DMR Heston model, and furthermore, it confirmed
with real market data that the calibration result utilizing the approximate
solution is analogous to the case using simulation methods. But the
approach can be applied to pricing VIX options only.

Our novel model is the first multi-factor affine diffusion model to
give practically feasible pricing formulas for both SPX options and
VIX options. To be concrete, we propose a new two-factor affine diffusion
model, and induce approximate quasi-closed formulas for both the options
based upon the perturbation theory. The analytic tractability greatly
improve the efficiency of calibration. As a result, we could achieve
satisfactory calibration results with long-term data for the SPX and
VIX from 2016 to 2018, consisting of hundreds of thousands of options.

The remainder of the paper is organized as follows. We introduce a
new two-factor model with multiscale volatility to jointly price VIX
options and SPX options in Section 2. We derive the analytic pricing
formulas of SPX options and VIX options under the proposed model using
asymptotic methods in Section 3. We present calibration with real
options data for the SPX and VIX in Section 4. We provide concluding
remarks in Section 5.

\section{A new two-factor multiscale stochastic volatility model}

Before introducing our model, we briefly review the double Heston
model of \citet{christoffersen2009shape}. The model was originally
devised to reflect stochastic correlation as well as stochastic volatility.
The improvement is achievable because the model has two factors. By
contrast, the Heston model cannot capture stochastic correlation because
it has only one factor. The double Heston model, however, cannot capture
the multiscale characteristics of a market, because it was not designed
for this purpose (see \citealp{gallant1999using,chernov2003alternative}).
Thus, the model needs to be extended to depict well-separated time
scales. More importantly, the model provides an analytic formula for
VIX options, which requires heavy computation for a double integral.
In particular, when practitioners want to fit the model to prices
for both SPX options and VIX options consistently within a limited
time, the time cost raises a severe problem.

To overcome these limitations of the double Heston model, we propose
a two-factor multiscale stochastic volatility model, which is given
by the following stochastic differential equations: 
\begin{align}
dX_{t} & =rX_{t}dt+\sqrt{Y_{t}}X_{t}dW_{t}^{X,1}+\sqrt{Z_{t}}X_{t}dW_{t}^{X,2},\nonumber \\
dY_{t} & =\frac{1}{\epsilon}\left(Z_{t}-Y_{t}\right)dt+\frac{\sqrt{2}\nu}{\sqrt{\epsilon}}\sqrt{Y_{t}}dW_{t}^{Y},\label{eq:model}\\
dZ_{t} & =\kappa\left(\theta-Z_{t}\right)dt+\sigma\sqrt{Z_{t}}dW_{t}^{Z}\nonumber 
\end{align}
for $0<\epsilon\ll1$ under a risk-neutral measure $\mathcal{Q}$,
where $r$ is the risk-free rate, and $\kappa\ll1/\epsilon$ is assumed.
The correlation structure of our model is 
\begin{gather*}
dW_{t}^{X,1}dW_{t}^{Y}=\eta dt,\quad dW_{t}^{X,2}dW_{t}^{Z}=\rho dt,\\
dW_{t}^{X,1}dW_{t}^{Z}=dW_{t}^{X,2}dW_{t}^{Y}=dW_{t}^{X,1}dW_{t}^{X,2}=dW_{t}^{Y}dW_{t}^{Z}=0.
\end{gather*}
It should be noted that $Y_{t}$ rapidly reverts to $Z_{t}$ under
our model, which is a major difference with the double Heston model.
Because the time scales are clearly separated due to condition $\kappa\ll1/\epsilon$,
the short-term variance $ourY_{t}+Z_{t}$ rapidly reverts to mid-term
variance $2Z_{t}$, and $2Z_{t}$ regresses to long-term variance
$2\theta$ at a relatively slow rate. On the other hand, one can consider
that our model is obtained by adding a fast mean-reverting factor
to Heston's volatility or otherwise by modifying the double Heston
model to express the well-separated time scales.

\section{Pricing SPX and VIX options}

\subsection{\label{subsec:An-approximate-analytic}An approximate analytic formula
for SPX options}

By the risk-neutral pricing principle, the price $P_{s}^{\epsilon}$
for an SPX call option with strike $K$ and maturity $T$ at time
$t$ is 
\begin{equation}
P_{s}^{\epsilon}\left(t,x,y,z\right)=e^{-r\left(T-t\right)}{\rm E}_{x,y,z}^{\mathcal{Q}}\left[\left(X_{T}-K\right)^{+}\right],\label{eq:spx_option_martingale_price}
\end{equation}
where $\left(X_{t},Y_{t},Z_{t}\right)=\left(x,y,z\right)$. Applying
the Feynman--Kac theorem to our model (\ref{eq:model}), it can be
shown that $P_{s}^{\epsilon}$ should satisfy the following partial
differential equation (PDE): 
\[
\left(\frac{1}{\epsilon}\mathcal{L}_{s,0}+\frac{1}{\sqrt{\epsilon}}\mathcal{L}_{s,1}+\mathcal{L}_{s,2}\right)P_{s}^{\epsilon}\left(t,x,y,z\right)=0,
\]
where 
\begin{eqnarray}
\mathcal{L}_{s,0} & = & \left(z-y\right)\partial_{y}+\nu^{2}y\partial_{yy}^{2},\nonumber \\
\mathcal{L}_{s,1} & = & \sqrt{2}\eta\nu xy\partial_{xy}^{2},\label{eq:spx_full_PDE}\\
\mathcal{L}_{s,2} & = & \partial_{t}+rx\partial_{x}+\kappa\left(\theta-z\right)\partial_{z}+\frac{1}{2}x^{2}\left(y+z\right)\partial_{xx}^{2}+\frac{1}{2}\sigma^{2}z\partial_{zz}^{2}+\rho\sigma xz\partial_{xz}^{2}-r\cdot\nonumber 
\end{eqnarray}
with the final condition $P_{s}^{\epsilon}\left(T,x,y,z\right)=\left(x-K\right)^{+}$.
Deriving an analytic solution for this PDE would be desirable; however,
we confirm that the solution cannot be easily induced.

Thus, we attempt to obtain an approximate solution for $P_{s}^{\epsilon}$
based on asymptotic methods (refer to \citet{fouque2011multiscale}).
Let us expand $P_{s}^{\epsilon}$ with respect to $\sqrt{\epsilon}$,
that is, $P_{s}^{\epsilon}=P_{s,0}+\sqrt{\epsilon}P_{s,1}+\epsilon P_{s,2}+\cdots.$
We wish to approximate $P_{s}^{\epsilon}$ to the sum of the leading
term $P_{s,0}$ and the first non-zero correction term $P_{s,1}$.
To do this, we show that $P_{s,0}$ and $P_{s,1}$ follow the following
PDEs independent of $y$: 
\begin{itemize}
\item leading term $P_{s,0}$ 
\begin{gather}
\left\langle \mathcal{L}_{s,2}\right\rangle P_{s,0}=0,\nonumber \\
P_{s,0}\left(T,x,z\right)=\left(x-K\right)^{+},\label{eq:PDE_p_s_0}
\end{gather}
\item first non-zero correction term $P_{s,1}$ 
\begin{gather}
\left\langle \mathcal{L}_{s,2}\right\rangle P_{s,1}=W_{3}^{\epsilon}zx\partial_{x}\left(x^{2}\partial_{xx}\right)P_{s,0},\nonumber \\
P_{s,1}\left(T,x,z\right)=0,\label{eq:PDE_p_s_1}
\end{gather}
\end{itemize}
where 
\begin{gather*}
\left\langle \mathcal{L}_{s,2}\right\rangle =\partial_{t}+rx\partial_{x}+\kappa\left(\theta-z\right)\partial_{z}+x^{2}z\partial_{xx}^{2}+\frac{1}{2}\sigma^{2}z\partial_{zz}^{2}+\rho\sigma xz\partial_{xz}^{2}-r,\\
W_{3}^{\epsilon}=-\frac{1}{\sqrt{2}}\eta\nu\sqrt{\epsilon}.
\end{gather*}
If setting $2z=\xi$, the operator $\left\langle \mathcal{L}_{s,2}\right\rangle $
changes to 
\[
\mathcal{L}_{H}=\partial_{t}+rx\partial_{x}+\kappa\left(2\theta-\xi\right)\partial_{\xi}+\frac{1}{2}x^{2}\xi\partial_{xx}^{2}+\frac{1}{2}\left(\sqrt{2}\sigma\right)^{2}\xi\partial_{\xi\xi}^{2}+\left(\frac{1}{\sqrt{2}}\rho\right)\left(\sqrt{2}\sigma\right)x\xi\partial_{x\xi}^{2}-r.
\]
The resulting operator $\mathcal{L}_{H}$ can be then regarded as
that generated from the Heston model whose parameters are the mean-reverting
rate $\kappa$, the long-term variance level $2\theta$, the volatility
of volatility $\sqrt{2}\sigma$, and the leverage correlation $\rho/\sqrt{2}$.

The analytic solutions for the PDEs can be obtained by exploiting
the analytical tractability of affine models fully. In fact, the key
idea for the derivation of the solutions is strongly motivated by
the work of \citet{fouque2011fast}. However, under their model, the
solutions for SPX options are involved in triple integrals requiring
overly large computational resources. By contrast, it is proved that
our solutions are all involved in single integrals. More importantly,
they give a solution only for SPX options, but we provide consistent
solutions for both SPX options and VIX options. The expressions for
$P_{s,0}$ and $P_{s,1}$ are summarized in the following theorem. 
\begin{thm}
\label{spx_theorem}The price $P_{s}^{\epsilon}$ in (\ref{eq:spx_option_martingale_price})
can be approximated with an accuracy of order $O\left(\epsilon\right)$,
that is, 
\[
\left|P_{s}^{\epsilon}-\left(P_{s,0}+P_{s,1}\right)\right|<C\epsilon
\]
for some $C>0$. Furthermore, $P_{s,0}$ and $P_{s,1}$ have the following
forms: 
\[
P_{s,0}\left(t,x,z\right)=\frac{e^{-r\tau}}{2\pi}\int e^{-ikq}\hat{G}\left(\tau,k,2z\right)\hat{h}\left(k\right)dk,
\]
\[
P_{s,1}\left(t,x,z\right)=\frac{e^{-r\tau}}{2\pi}\int e^{-ikq}b\left(k\right)\left(\kappa\theta\hat{f}_{0}\left(\tau,k\right)+2z\hat{f}_{1}\left(\tau,k\right)\right)\hat{G}\left(\tau,k,2z\right)\hat{h}\left(k\right)dk,
\]
where $\tau=T-t$, $q=r\tau+\log x$, 
\begin{align*}
\hat{G}\left(\tau,k,y\right) & =e^{C\left(\tau,k\right)+yD\left(\tau,k\right)},\\
\hat{h}\left(k\right) & =\frac{K^{1+ik}}{ik-k^{2}},\\
\hat{f}_{0}\left(\tau,k\right) & =\frac{2\tau d\left(k\right)g\left(k\right)+g^{2}\left(k\right)-1}{d^{2}\left(k\right)g\left(k\right)\left(g\left(k\right)e^{\tau d\left(k\right)}-1\right)}+\frac{\tau d\left(k\right)g\left(k\right)-g\left(k\right)-1}{d^{2}\left(k\right)g\left(k\right)},\\
\hat{f}_{1}\left(\tau,k\right) & =\frac{e^{\tau d\left(k\right)}\left(g\left(k\right)^{2}\left(e^{\tau d\left(k\right)}-1\right)-2\tau d\left(k\right)g\left(k\right)+1\right)-1}{d\left(k\right)\left(g\left(k\right)e^{\tau d\left(k\right)}-1\right)^{2}},
\end{align*}
and 
\begin{align*}
C\left(\tau,k\right) & =\frac{\kappa\theta}{\sigma^{2}}\left(\left(\kappa+\rho ik\sigma-d\left(k\right)\right)\tau-2\log\left(\frac{e^{-\tau d\left(k\right)}-g\left(k\right)}{1-e^{-\tau d\left(k\right)}g\left(k\right)}\right)\right),\\
D\left(\tau,k\right) & =\frac{\kappa+\rho ik\sigma+d\left(k\right)}{2\sigma^{2}}\left(\frac{e^{-\tau d\left(k\right)}-1}{e^{-\tau d\left(k\right)}-g\left(k\right)}\right),\\
b\left(k\right) & =-\frac{1}{2}W_{3}^{\epsilon}\left(ik^{3}+k^{2}\right),\\
d\left(k\right) & =\sqrt{2\sigma^{2}\left(k^{2}-ik\right)+\left(\kappa+\rho ik\sigma\right)^{2}},\\
g\left(k\right) & =\frac{\kappa+\rho ik\sigma+d\left(k\right)}{\kappa+\rho ik\sigma-d\left(k\right)}.
\end{align*}
\end{thm}

\begin{proof}
The detailed proof is in \ref{sec:appendix_a}. 
\end{proof}

\subsection{Relationship between the VIX and our model}

Before proceeding, we show the relationship between the VIX\footnote{http://www.cboe.com/micro/vix/vixwhite.pdf}
and our model. Because the VIX at time $t$, ${\rm VIX}_{t}^{\epsilon}$,
is the square root of an integration of spot variance from $t$ to
$t+\tau_{0}$ $\left(\tau_{0}:=30/365\right)$, based on \citet{carr1998towards},
the relationship can be given by 
\begin{align}
\left({\rm VIX}_{t}^{\epsilon}\right)^{2} & =100^{2}\times{\rm E}^{\mathcal{Q}}\left[\frac{1}{\tau_{0}}\int_{t}^{t+\tau_{0}}\left(Y_{s}+Z_{s}\right)ds\right]\nonumber \\
 & =100^{2}\times\frac{1}{\tau_{0}}\int_{t}^{t+\tau_{0}}\left({\rm E}^{\mathcal{Q}}\left[Y_{s}\right]+{\rm E}^{\mathcal{Q}}\left[Z_{s}\right]\right)ds\label{eq:relation_ours}\\
 & =100^{2}\times\left(a_{1}^{\epsilon}Y_{t}+a_{2}^{\epsilon}Z_{t}+\left(a_{3}^{\epsilon}+a_{4}^{\epsilon}\right)\theta\right),\nonumber 
\end{align}
where 
\begin{align*}
a_{1}^{\epsilon} & =\frac{\epsilon}{\tau_{0}}\left(1-e^{-\tau_{0}/\epsilon}\right),\\
a_{2}^{\epsilon} & =\frac{1}{\kappa\tau_{0}}\left(1-e^{-\kappa\tau_{0}}\right)+\frac{1}{1-\kappa\epsilon}\left(\frac{1}{\kappa\tau_{0}}\left(1-e^{-\kappa\tau_{0}}\right)-\frac{\epsilon}{\tau_{0}}\left(1-e^{-\tau_{0}/\epsilon}\right)\right),\\
a_{3}^{\epsilon} & =1-a_{1}^{\epsilon},\\
a_{4}^{\epsilon} & =1-a_{2}^{\epsilon}.
\end{align*}
The following facts induced by Ito's lemma are used for the derivation
of the above formula \eqref{eq:relation_ours}. 
\begin{align*}
\text{E}\left[Y_{s}\right] & =e^{-\frac{1}{\epsilon}\left(s-t\right)}Y_{t}+\frac{1}{1-\kappa\epsilon}\left(e^{-\kappa\left(s-t\right)}-e^{-\frac{1}{\epsilon}\left(s-t\right)}\right)Z_{t}+\left[\left(1-e^{-\frac{1}{\epsilon}\left(s-t\right)}\right)-\frac{1}{1-\kappa\epsilon}\left(e^{-\kappa\left(s-t\right)}-e^{-\frac{1}{\epsilon}\left(s-t\right)}\right)\right]\theta,\\
\text{E}\left[Z_{s}\right] & =e^{-\kappa\left(s-t\right)}Z_{t}+\theta\left(1-e^{-\kappa\left(s-t\right)}\right).
\end{align*}
It is noteworthy that $\left({\rm VIX}_{t}^{\epsilon}\right)^{2}$
converges to 
\[
\left({\rm VIX}_{t}^{*}\right)^{2}=100^{2}\times\left(a_{2}^{*}Z_{t}+\left(1+a_{4}^{*}\right)\theta\right),
\]
as $\epsilon\rightarrow0$, where $a_{2}^{*}=\frac{2}{\kappa\tau_{0}}\left(1-e^{-\kappa\tau_{0}}\right)$
and $a_{4}^{*}=1-a_{2}^{*}$. In addition, ${\rm VIX}_{t}$ and ${\rm VIX}_{t}^{*}$
have the following relationship: 
\begin{align*}
\left({\rm VIX}_{t}^{\epsilon}\right)^{2} & -\left({\rm VIX}_{t}^{*}\right)^{2}=100^{2}\times\left(\frac{\epsilon}{\tau_{0}}\left(1-e^{-\tau_{0}/\epsilon}\right)\left(Y_{t}-Z_{t}\right)+\frac{\epsilon}{\tau_{0}}\left(1-e^{-\kappa\tau_{0}}\right)\left(Z_{t}-\theta\right)\right)+O\left(\epsilon^{2}\right):=\Delta{\rm VIX}_{t}^{2}.
\end{align*}

\subsection{An approximate analytic formula for VIX options}

Denoting the price for a VIX call option with strike $K$ and maturity
$T$ as $P_{v}^{\epsilon}$ at time $t$, $P_{v}^{\epsilon}$ is expressed
by 
\begin{align}
P_{v}^{\epsilon}\left(t,y,z\right) & =e^{-r\left(T-t\right)}{\rm E}_{y,z}^{\mathcal{Q}}\left[h^{\epsilon}\left(Y_{T},Z_{T}\right)\right],\label{eq:vix_option_martingale_price}
\end{align}
where $\left(Y_{t},Z_{t}\right)=\left(y,z\right)$, and $h^{\epsilon}\left(u,v\right)=\left({\rm VIX}_{T}^{\epsilon}\left(u,v\right)-K\right)^{+}$.
If $H\left(x\right):=\left(\sqrt{x}-K\right)^{+}$, $h^{\epsilon}\left(u,v\right)=H(\left({\rm VIX}_{T}^{\epsilon}\right)^{2})$.
Then, by expanding $H$ at $\left({\rm VIX}_{t}^{*}\right)^{2}$,
$h^{\epsilon}$ can be transformed to 
\begin{align*}
h^{\epsilon}\left(u,v\right) & =H(\left({\rm VIX}_{T}^{*}\right)^{2})+H'(\left({\rm VIX}_{T}^{*}\right)^{2})\Delta{\rm VIX}_{t}^{2}+\frac{1}{2}H''(\left({\rm VIX}_{T}^{*}\right)^{2})\left(\Delta{\rm VIX}_{t}^{2}\right)^{2}+\cdots\\
 & =(\sqrt{a_{2}^{*}v+\left(1+a_{4}^{*}\right)\theta}\times100-K)\mathbf{1}_{\left\{ v\ge\left(K^{2}-\left(1+a_{4}^{*}\right)\theta\right)/a_{2}^{*}\right\} }\\
 & \qquad+\epsilon\left(\frac{\left(1-e^{-\tau_{0}/\epsilon}\right)\left(u-v\right)+\left(1-e^{-\kappa\tau_{0}}\right)\left(v-\theta\right)}{2\tau_{0}\sqrt{a_{2}^{*}v+\left(1+a_{4}^{*}\right)\theta}}\right)\mathbf{1}_{\left\{ v\ge\left(K^{2}-\left(1+a_{4}^{*}\right)\theta\right)/a_{2}^{*}\right\} }\times100+O\left(\epsilon^{2}\right)\\
 & =h_{0}\left(v\right)+\epsilon h_{1}\left(u,v\right)+O\left(\epsilon^{2}\right),
\end{align*}
where $h_{0}$ and $h_{1}$ are the first and second terms of the
second line, respectively. Meanwhile, utilizing Ito's lemma, we can
show 
\begin{align*}
Y_{s} & =Z_{s}+e^{-\frac{1}{\epsilon}\left(s-t\right)}\left(Y_{t}-Z_{t}\right)+\frac{\sqrt{2}\nu}{\sqrt{\epsilon}}\int_{t}^{s}e^{-\frac{1}{\epsilon}\left(s-u\right)}\sqrt{Y_{u}}dW_{u}^{Y}\\
 & \qquad+\sum_{n=1}^{\infty}\left(\kappa\epsilon\right)^{n}\left[\left(Z_{s}-\theta\right)-e^{-\frac{1}{\epsilon}\left(s-t\right)}\left(Z_{t}-\theta\right)\right]-\sigma\sum_{n=0}^{\infty}\left(\kappa\epsilon\right)^{n}\int_{t}^{s}e^{-\frac{1}{\epsilon}\left(s-u\right)}\sqrt{Z_{u}}dW_{u}^{Z}.
\end{align*}
Using the tower property of conditional expectation and the linearity
of $h_{1}$ with respect to $u$, we obtain 
\begin{align*}
{\rm E}_{y,z}^{\mathcal{Q}}\left[h\left(Y_{T},Z_{T}\right)\right] & ={\rm E}_{y,z}^{\mathcal{Q}}\left[{\rm E}_{y,z}^{\mathcal{Q}}\left[h\left(Y_{T},Z_{T}\right)\left|Z_{T}\right.\right]\right]\\
 & ={\rm E}_{y,z}^{\mathcal{Q}}\left[{\rm E}_{y,z}^{\mathcal{Q}}\left[h_{0}\left(Z_{T}\right)+\epsilon h_{1}\left(Y_{T},Z_{T}\right)\left|Z_{T}\right.\right]\right]+O\left(\epsilon^{2}\right)\\
 & ={\rm E}_{y,z}^{\mathcal{Q}}\left[h_{0}\left(Z_{T}\right)+\epsilon h_{1}\left({\rm E}_{y,z}^{\mathcal{Q}}\left[Y_{T}\left|Z_{T}\right.\right],Z_{T}\right)\right]+O\left(\epsilon^{2}\right)\\
 & ={\rm E}_{y,z}^{\mathcal{Q}}\left[h_{0}\left(Z_{T}\right)+\epsilon h_{1}\left(Z_{T}+e^{-\left(T-t\right)/\epsilon}\left(y-z\right)+O\left(\epsilon\right),Z_{T}\right)\right]+O\left(\epsilon^{2}\right)\\
 & ={\rm E}_{y,z}^{\mathcal{Q}}\left[h_{0}\left(Z_{T}\right)+h_{1}^{*}\left(Z_{T}\right)\right]+O\left(\epsilon^{2}\right),
\end{align*}
where $h_{1}^{*}\left(v\right)=\epsilon h_{1}\left(v+e^{-\left(T-t\right)/\epsilon}\left(y-z\right),v\right)$.
Then, the facts presented so far lead to the following theorem.
\begin{thm}
\label{vix_theorem}The price $P_{v}^{\epsilon}$ in (\ref{eq:vix_option_martingale_price})
can be approximated with accuracy of order $O(\epsilon^{2})$, that
is, 
\[
\left|P_{v}^{\epsilon}-\left(P_{v,0}+P_{v,1}\right)\right|<C\epsilon^{2}.
\]
for some $C>0$. Furthermore, $P_{v,0}$ and $P_{v,1}$ are given
by 
\begin{align*}
P_{v,0}\left(t,z\right) & =e^{-rT}\int_{0}^{\infty}h_{0}(\delta\zeta)\,f(\zeta;k,\lambda)\,d\zeta,\\
P_{v,1}\left(t,y,z\right) & =e^{-rT}\int_{0}^{\infty}h_{1}^{*}(\delta\zeta)\,f(\zeta;k,\lambda)\,d\zeta,
\end{align*}
where $f\left(\zeta;k,\lambda\right)$ is the density function of
a non-central chi-square distribution with $k=\frac{4\kappa\theta}{\sigma^{2}}$
degrees of freedom and non-centrality parameter $\lambda=\frac{ze^{-\kappa\tau}}{\delta}$
with $\tau=T-t$, where $\delta=\frac{\left(1-e^{-\kappa\tau}\right)\sigma^{2}}{4\kappa}$,
and 
\begin{align*}
h_{0}\left(v\right) & =(\sqrt{a_{2}^{*}v+\left(1+a_{4}^{*}\right)\theta}\times100-K)^{+}\mathbf{1}_{\left\{ v\ge\left(K^{2}-\left(1+a_{4}^{*}\right)\theta\right)/a_{2}^{*}\right\} },\\
h_{1}^{*}\left(v\right) & =\epsilon\left(\frac{2e^{-\tau/\epsilon}a_{1}^{\epsilon}\left(y-z\right)+\kappa\epsilon a_{2}^{*}\left(v-\theta\right)}{4\sqrt{a_{2}^{*}v+\left(1+a_{4}^{*}\right)\theta}}\right)\mathbf{1}_{\left\{ v\ge\left(K^{2}-\left(1+a_{4}^{*}\right)\theta\right)/a_{2}^{*}\right\} }\times100.
\end{align*}
\end{thm}

\begin{proof}
Because it is known that future values of the Cox--Ingersoll--Ross
(CIR) process $Z_{t}$ follow a non-central chi-squared distribution,
the values of $P_{v,0}$ and $P_{v,1}$ can be expressed as the expectations
with respect to the distribution \citep[c.f.][]{brigo2007interest}. 
\end{proof}
\noindent Note that the above analytic formulas in Theorem \ref{vix_theorem}
are involved in single Integrals, as in the case of the formulas for
SPX options. This means that it is possible to compute the values
without difficulty. We also strongly stress that our model has only
six parameters $\kappa$, $\theta$, $\sigma$, $\rho$, $\epsilon$,
and $W_{3}^{\epsilon}$, which is only two more than the Heston model.
Some of the parameters $\left(W_{3}^{\epsilon},\rho\right)$ and $\epsilon$
control only of SPX and VIX option prices, respectively.

\section{Empirical test}

In this section, we verify the outperformance of our model through
an empirical test using real market data. To show the benefits of
two-factor models clearly, we choose the Heston model, which is the
most well-known one-factor model, as a benchmark. While it may be
considered more reasonable to compare our model with another two-factor
model, there is no existing two-factor model that gives analytic pricing
formulas that are available and efficient for both SPX and VIX options.
For example, the double Heston model \citep{christoffersen2009shape}
provides a double integral form for the pricing of VIX options but
cannot be considered in practice because it takes a huge amount of
time to compute, as its parameters are found based on a lot of option
prices. 
\begin{figure}
\centering{}\includegraphics[scale=0.65]{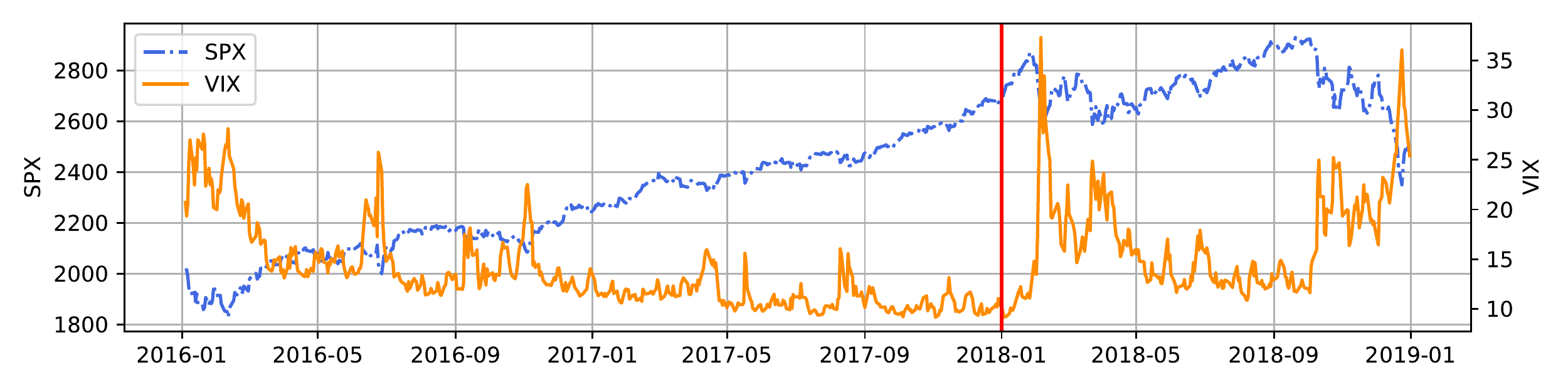}\caption{\label{fig:spx_vix}Series of the SPX and the VIX for 2016--2018.
The vertical red line indicates the dividing date, the left and the
right of which correspond to the training period (2016--2017) and
the test period (2018), respectively.}
\end{figure}

\begin{table}
\centering{}%
\begin{tabular}{cccccccccc}
\toprule 
 & \multicolumn{4}{c}{SPX} &  & \multicolumn{4}{c}{VIX}\tabularnewline
\midrule 
 & mean & deviation & skewness & kurtosis &  & mean & deviation & skewness & kurtosis\tabularnewline
\midrule
\midrule 
training & 0.001 & 0.007 & -0.530 & 4.379 &  & -0.001 & 0.073 & 0.816 & 6.230\tabularnewline
\midrule 
test & 0.000 & 0.011 & -0.435 & 3.319 &  & 0.004 & 0.100 & 2.170 & 13.289\tabularnewline
\midrule 
all & 0.000 & 0.008 & -0.581 & 5.295 &  & 0.000 & 0.083 & 1.669 & 12.385\tabularnewline
\bottomrule
\end{tabular}\caption{\label{tab:moments}Moments for the log-return series of the SPX and
the VIX for the data period.}
\end{table}

We obtain options data on the SPX and the VIX from 2016 to 2018 from
the CBOE data shop\footnote{https://datashop.cboe.com/}. We filter
them out if their daily trades are below 50, if their prices do not
reach 0.5 point, or if they expire within 3 days. We then split the
data into a training set for 2016--2017 (235,227 SPX options and
46,516 VIX options) and a test set for 2018 (176,866 SPX options and
23,347 VIX options). Figure \ref{fig:spx_vix} shows the series of
the SPX and the VIX for the data period. The vertical red line indicates
the dividing date, the left and the right of which correspond to the
training period and the test period, respectively. Many financial
crises arose during the period: Chinese stock market turbulence (early
2016), Brexit (late 2016), and US stock market downturn (2018). It
is clear that the SPX falls sharply and the VIX increases dramatically
during each crisis. Table \ref{tab:moments} summarizes the moments
for the log-returns of both series for the data period. The VIX returns
move in a much more volatile way than the SPX returns move. In addition,
the SPX returns show negative skewness and weak kurtosis ($>3$),
while the VIX returns show positive skewness and strong kurtosis ($\gg3$).
These findings imply that VIX options tend to be more expensive and
riskier than SPX options are.

We adopt a two-step approach for calibrations for the Heston model
and our model. Specifically, we perform a step-by-step calibration,
in which the prices for specific kinds of products (e.g., SPX option)
among all price data are used for parameter estimation for each step
only. Such an approach is often observed in various studies in the
literature (\citealp{bayer2013fast,goutte2017regime}). We first attempt
to find the parameters $\kappa$, $\theta$, $\sigma$, and $\epsilon$
to control VIX option prices, and the other parameters $\rho$, $W_{3}^{\epsilon}$
are sought with SPX option prices. Note that the VIX option prices
may be more sensitive to the parameters $\kappa$, $\theta$, and
$\sigma$ than the SPX option prices. Because a highly precise optimization
needs highly sensitive Greeks, it seems reasonable to prefer VIX
options to SPX options for an estimation of the parameters $\kappa$,
$\theta$, and $\sigma$. We first define the method for the Heston
model with the four parameters $\kappa,\theta,\sigma$, and $\rho$
as follows. 
\begin{enumerate}
\item First, perform calibration with VIX options 
\[
\underset{\kappa,\theta,\sigma}{{\rm argmin}}\sum_{i\in I}\left(\sum_{j\in J_{v,i}}\left(\frac{P_{v}^{mdl}\left(z_{i},\kappa,\theta,\sigma;t_{i,}K_{j},T_{j}\right)-P_{v}^{mkt}\left(t_{i,},K_{j},T_{j}\right)}{0.1+P_{v}^{mkt}\left(t_{i,},K_{j},T_{j}\right)}\right)^{2}\right)
\]
such that 
\begin{equation}
{\rm VIX}_{i}^{mkt}=\sqrt{b_{2}^{*}z_{i}+b_{4}^{*}\theta}.\label{eq:VIX_relation_heston}
\end{equation}
where $b_{2}^{*}=a_{2}^{*}/2$ and $b_{4}^{*}=\left(1+a_{4}^{*}\right)/2$. 
\item Second, perform calibration with SPX options 
\end{enumerate}
\[
\underset{\rho}{{\rm argmin}}\sum_{i\in I}\left(\sum_{j\in J_{s,i}}\left(\frac{P_{s}^{mdl}\left(\rho;t_{i,}K_{j},T_{j};z_{i},\kappa,\theta,\sigma\right)-P_{s}^{mkt}\left(t_{i,},K_{j},T_{j}\right)}{0.1+P_{s}^{mkt}\left(t_{i,},K_{j},T_{j}\right)}\right)^{2}\right).
\]
The superscripts $mkt$ and $mdl$ are used to indicate the market
price and the model price for an option, respectively. $I$ is the
index set of the dates for the training set. $J_{v,i}$ and $J_{s,i}$
are the index sets of the VIX options and the SPX options, respectively,
on the $i$th date. We explain the two-step method for our model with
the six parameters $\kappa,\theta,\sigma,\rho,\epsilon$, and $W_{3}^{\epsilon}$
as follows. 
\begin{enumerate}
\item First, perform calibration with VIX options 
\[
\underset{\kappa,\theta,\sigma,\epsilon}{{\rm argmin}}\sum_{i\in I}\left(\underset{y_{i},z_{i}}{{\rm argmin}}\sum_{j\in J_{v,i}}\left(\frac{P_{v}^{mdl}\left(y_{i},z_{i},\kappa,\theta,\sigma,\epsilon;t_{i,}K_{j},T_{j}\right)-P_{v}^{mkt}\left(t_{i,},K_{j},T_{j}\right)}{0.1+P_{v}^{mkt}\left(t_{i,},K_{j},T_{j}\right)}\right)^{2}\right)
\]
such that 
\begin{equation}
{\rm VIX}_{i}^{mkt}=\sqrt{a_{1}^{\epsilon}y_{i}+a_{2}^{\epsilon}z_{i}+\left(a_{3}^{\epsilon}+a_{4}^{\epsilon}\right)\theta}.\label{eq:VIX_relation_ours}
\end{equation}
\item Second, perform calibration with SPX options 
\end{enumerate}
\[
\underset{\rho,W_{3}^{\epsilon}}{{\rm argmin}}\sum_{i\in I}\left(\sum_{j\in J_{s,i}}\left(\frac{P_{s}^{mdl}\left(\rho,W_{3}^{\epsilon};t_{i,}K_{j},T_{j};y_{i},z_{i},\kappa,\theta,\sigma\right)-P_{s}^{mkt}\left(t_{i,},K_{j},T_{j}\right)}{0.1+P_{s}^{mkt}\left(t_{i,},K_{j},T_{j}\right)}\right)^{2}\right).
\]
Note that, under the Heston model, $z_{i}$ is uniquely determined
by the relationship (\ref{eq:VIX_relation_heston}) between the model
and the VIX. By contrast, under our model, $y_{i}$ and $z_{i}$ cannot
be uniquely determined by the relationship (\ref{eq:VIX_relation_ours}),
which is why we also minimize the daily objective with respect to
$y_{i}$ and $z_{i}$ for each $i$th date. Because this process requires
large computational resources, the calibrations for both models are
implemented based upon C++ and OpenMP, and executed in parallel on
24 CPU cores of two Intel Xeon Gold 5118. As a result, we obtain the
following parameters for each model: $\kappa=3.43$, $\theta=0.0400$,
$\sigma=0.424$, and $\rho=-1$ for the Heston model and $\kappa=3.58$,
$\theta=0.021$ $\sigma=0.347$, $\rho=-1$, $\epsilon=0.0096$, and
$W_{3}^{\epsilon}=0.0150$ for our model. Note that the two time scales
for our model are well separated, as $\kappa=3.58$ and $1/\epsilon=104.17$.
\begin{figure}[t]
\centering{}\subfloat[errors for SPX options]{\begin{centering}
\includegraphics[scale=0.65]{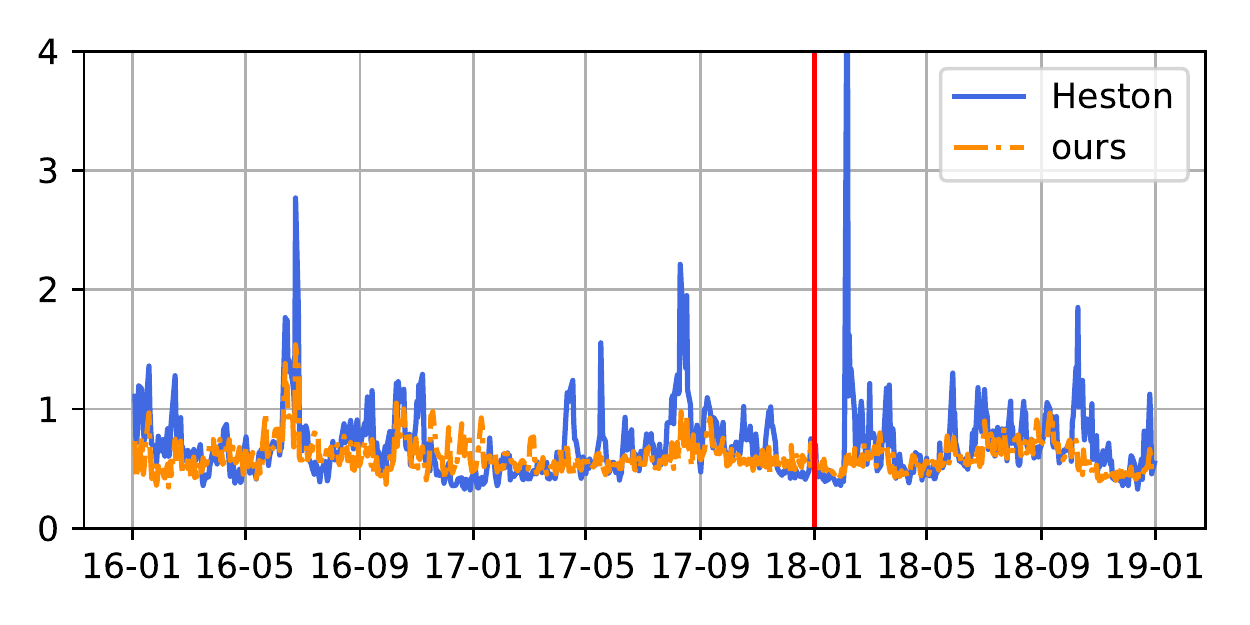}
\par\end{centering}
}\subfloat[errors for VIX options]{\begin{centering}
\includegraphics[scale=0.65]{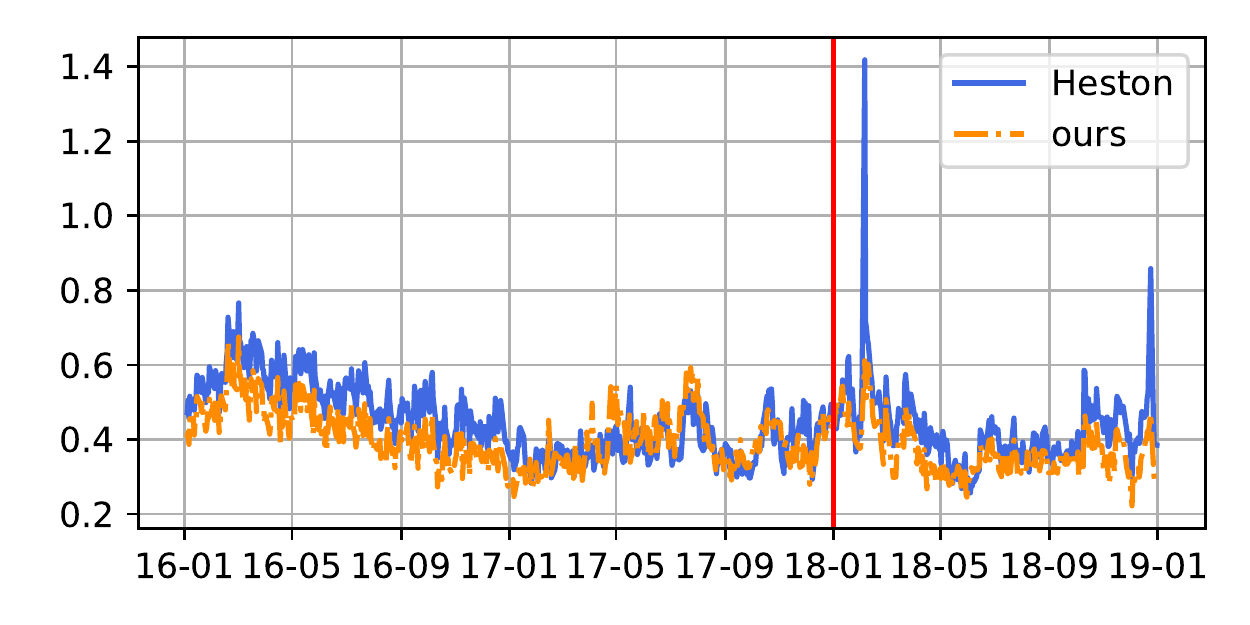}
\par\end{centering}
}\caption{\label{fig:daily-errors}Comparison of daily errors of both models
by product types for 2016--2018. The vertical red line indicates
the dividing date, the left and the right of which are the training
period (2016--2017) and the test period (2018), respectively.}
\end{figure}

\begin{figure}[t]
\centering{}\subfloat[$\sqrt{z_{i}}$ for the Heston model]{\begin{centering}
\includegraphics[scale=0.65]{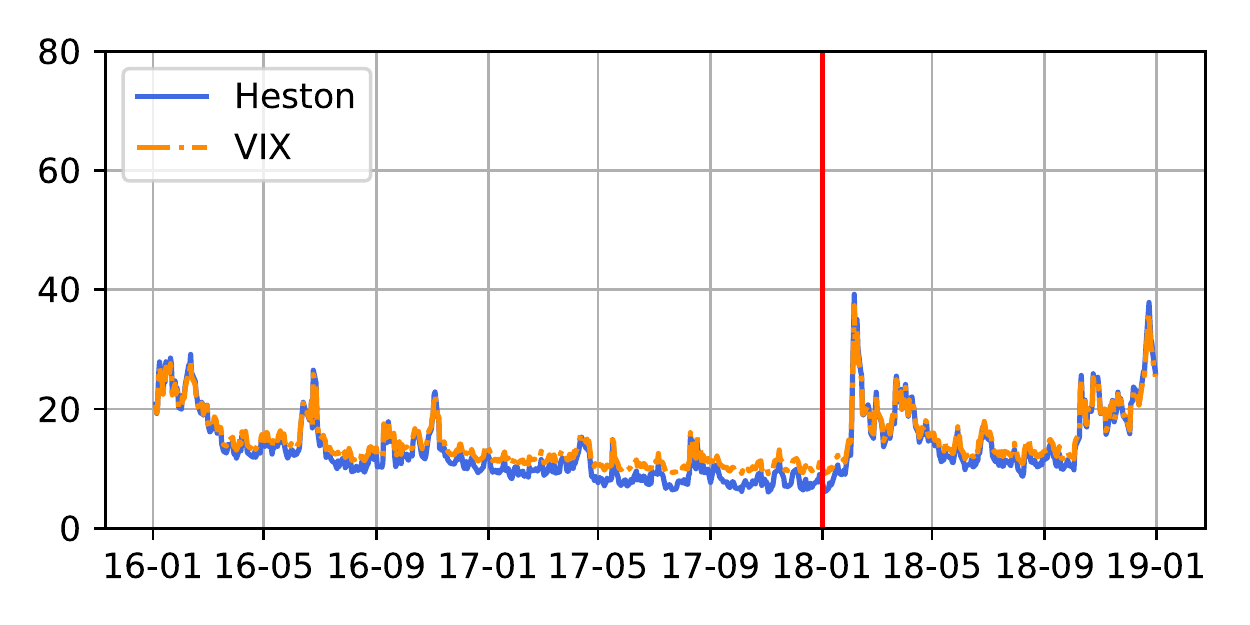}
\par\end{centering}
}\subfloat[$\sqrt{y_{i}+z_{i}}$ for our model]{\begin{centering}
\includegraphics[scale=0.65]{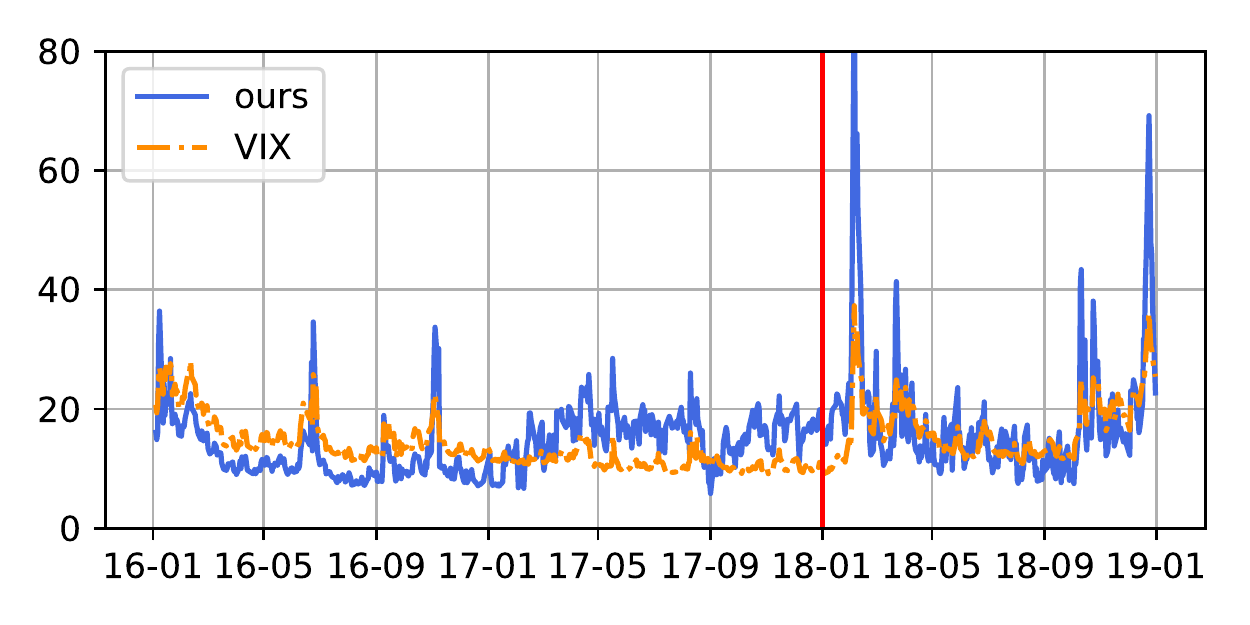}
\par\end{centering}
}\caption{\label{fig:spot-volatility}Spot volatilities of both models for 2016--2018.
The vertical red line indicates the dividing date.}
\end{figure}

Figure \ref{fig:daily-errors} compares the daily errors of the Heston
model and our model for the data period. The left and right figures
indicate the errors for the SPX options and the VIX options, respectively.
As before, the vertical red line represents the dividing date for
the training data and the test data, respectively. At first glance,
our model produces less errors, and it is more robust than the Heston
model. Precisely, our model reduces the errors on the training sets
of the SPX and the VIX options by 9.9\% and 13.2\%, respectively,
and decreases the errors on the test sets of the SPX and the VIX options
by 13.0\% and 16.5\%, respectively, compared with the Heston model
(refer to \ref{sec:errors_of_test}). In particular, the Heston model
gives fairly large errors for particular dates. We can postulate a
plausible reason for those phenomena from Figure \ref{fig:spot-volatility},
which draws the spot volatilities of both models for the data period.
In the figure, the hidden process for our model appears to be more
dynamic than that for the Heston model. This is because the spot volatility
$\sqrt{z_{i}}$ for the Heston model is strongly correlated with the
VIX but $\sqrt{y_{i}+z_{i}}$ for our model is not. Our model captures
short-term volatility that is difficult to detect and produces a more
volatile process. If our model is assumed to be correct, we conclude
that the spot volatilities for the Heston model are fairly often biased,
especially when sudden strong shocks impact the market. The bias eventually
results in large fitting or prediction errors of the Heston model
for short-term products, as shown in Table \ref{sec:errors_of_test}.
The table sums up the errors separately by time to maturity. The table
results confirm that our model performs better than the Heston model
as the time to maturity for an option becomes shorter, provided that
the time to expiration is not too short. Note that the asymptotic
method is valid when the time to maturity is not too short.

\begin{figure}[t]
\centering{}\subfloat[$y=0.0234$, $z=0.0194$]{\begin{centering}
\includegraphics[scale=0.63]{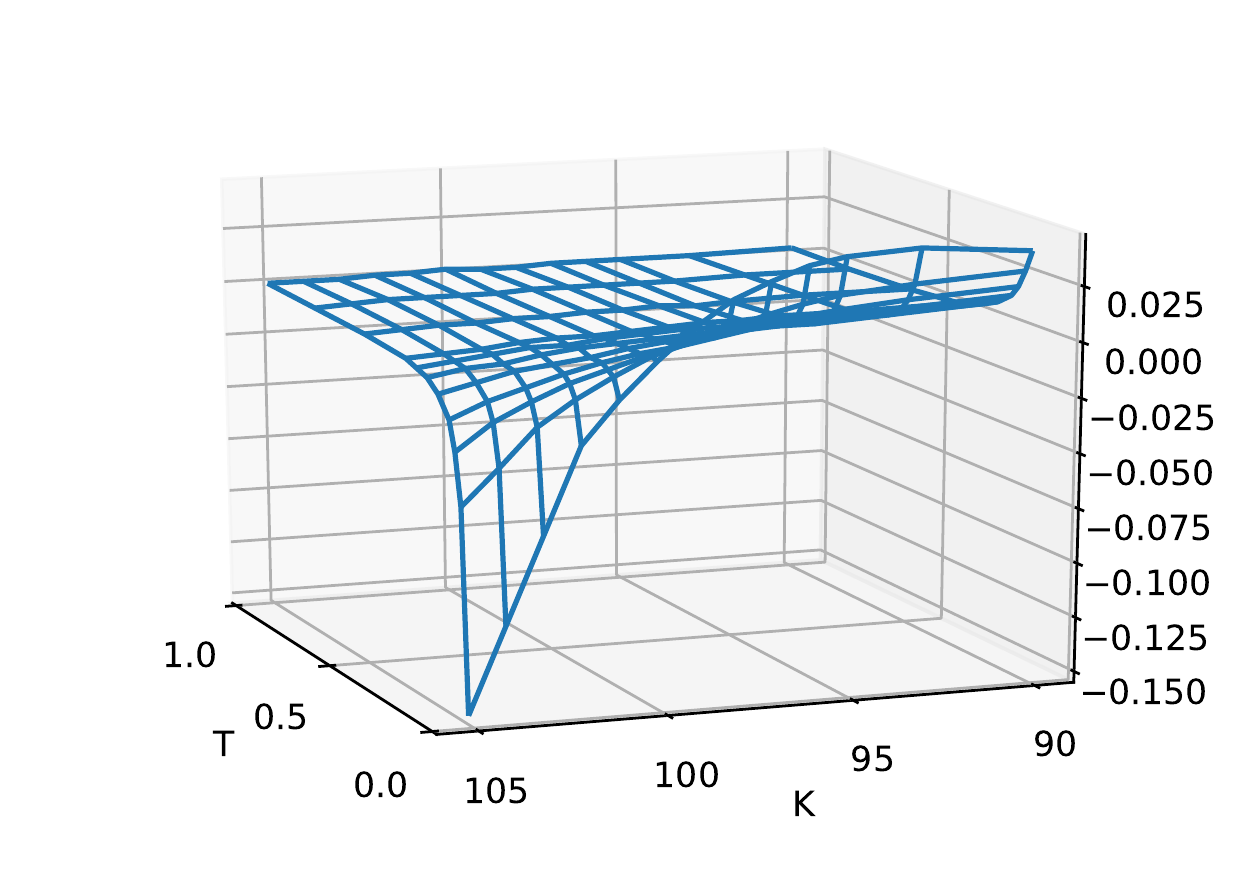}
\par\end{centering}
}\subfloat[$y=0.0110$, $z=0.0203$]{\begin{centering}
\includegraphics[scale=0.63]{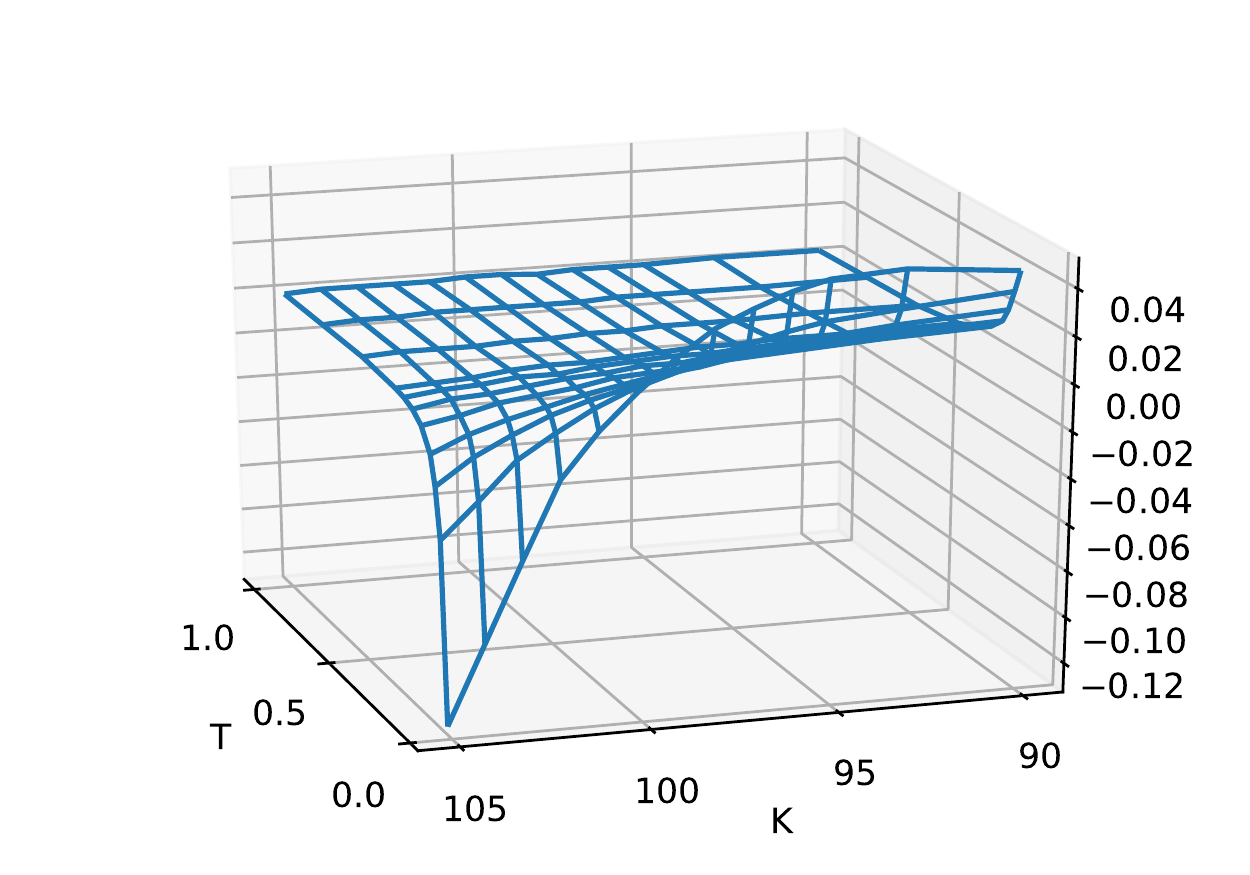}
\par\end{centering}
}\caption{\label{fig:SPX-imvol}Two differences between the SPX implied volatility
surfaces from the corrected and uncorrected prices. Here, $z=0.0197$
for the uncorrected case.}
\end{figure}

\begin{figure}[t]
\centering{}\subfloat[$y=0.0234$, $z=0.0194$]{\begin{centering}
\includegraphics[scale=0.63]{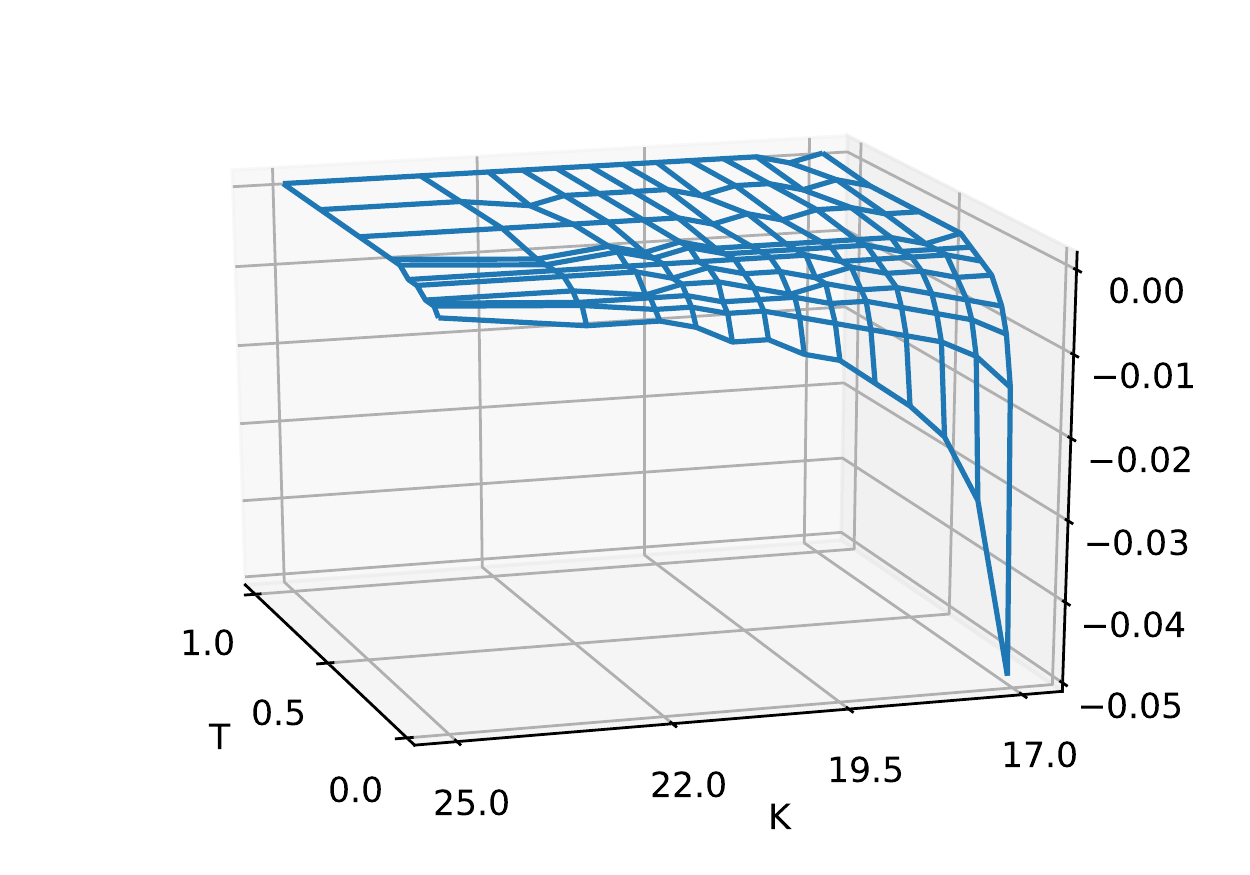}
\par\end{centering}
}\subfloat[$y=0.0110$, $z=0.0203$]{\begin{centering}
\includegraphics[scale=0.63]{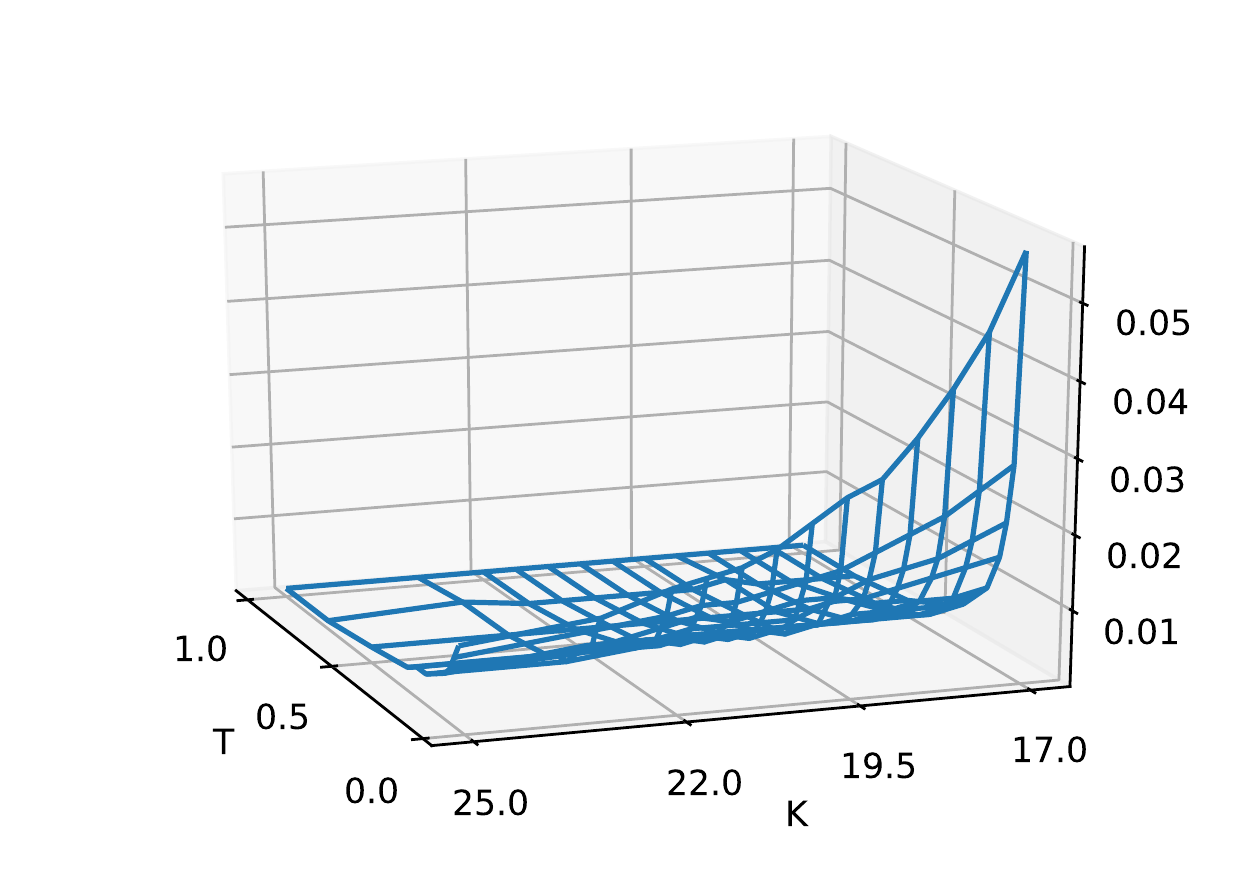}
\par\end{centering}
}\caption{\label{fig:VIX-imvol}Two differences between the VIX implied volatility
surfaces from the corrected and uncorrected prices.}
\end{figure}

Furthermore, implied volatility surfaces for SPX options and VIX options
clearly show how appropriate our model is for fitting to short-term
products. To define the implied volatility for VIX options consistently
, we first consider the following simple model: 
\[
d{\rm VIX}_{t}=\sigma dW_{t}.
\]
In other words, ${\rm VIX}_{T}$ follows the normal distribution $N\left({\rm VIX}_{t},\sigma\sqrt{T-t}\right)$
under the model. Then, it is easy to show that the price $c$ of a
VIX call option with maturity $T$ and strike $K$ is given by 
\[
c\left({\rm VIX}_{t}\right)=\left({\rm VIX}_{t}-K\right)N\left(\frac{{\rm VIX}_{t}-K}{\sigma\sqrt{T-t}}\right)+\sigma\sqrt{T-t}\ n\left(\frac{{\rm VIX}_{t}-K}{\sigma\sqrt{T-t}}\right)
\]
where $N$ and $n$ are the cumulative distribution function and the
probability density function of the unit Gaussian distribution, respectively.
As a result, for the market price $c^{mkt}$, the corresponding implied
volatility for the VIX option can be defined as $\sigma^{imp}$ to
match the model price with $c^{mkt}$.

We now generate implied volatility surfaces $\sigma_{0,1}^{imp}$
from the corrected prices $P_{s,0}\left(z\right)+P_{s,1}\left(z\right)$
and $P_{v,0}\left(z\right)+P_{v,1}\left(y,z\right)$ with the estimated
parameters $\kappa=3.58$, $\theta=0.021$ $\sigma=0.347$, $\rho=-1$,
$\epsilon=0.0096$, and $W_{3}^{\epsilon}=0.0150$ for two hidden
states $\left(y,z\right)=\left(0.0234,0.0194\right)$ and $\left(y,z\right)=\left(0.0110,0.0203\right)$.
We also generate an additional implied volatility surface $\sigma_{0}^{imp}$
from the uncorrected prices $P_{s,0}\left(z\right)$ and $P_{v,0}\left(z\right)$
for hidden state $z=0.0197$, which is equivalent to generating the
surface using the corrected prices with the same parameters as the
preceding case but $\epsilon=0$ and $W_{3}^{\epsilon}=0$. All the
hidden states $y$, $z$ are chosen so that the model value of the
VIX is $20$, that is, ${\rm VIX}_{t}=20$ in the relationship \eqref{eq:relation_ours}.

Figures \ref{fig:SPX-imvol} and \ref{fig:VIX-imvol} present the
differences $\sigma_{0,1}^{imp}-\sigma_{0}^{imp}$ for the SPX cases
and the VIX cases, respectively. The figures confirm once again that
the corrected terms $P_{s,1}$ and $P_{v,1}$ allow excellent short-term
flexibility, and thereby capture the fluctuations in the slope and
the level of the volatility smirk. Based on Figure \ref{fig:SPX-imvol},
$P_{s,1}$ may make short-term in-the-money (SITM) options more expensive
than short-term out-of-the-money (SOTM) options in a somewhat robust
way. However, based on Figure \ref{fig:VIX-imvol}, a subtle change
of $z$ in our model seems to lead to a significantly different short-term
VIX market. When the short-term state $y$ is smaller than the mid-term
state $z$, as in the right figure, $P_{v,1}$ makes the SITM options
more expensive than the SOTM options, as in the SPX cases, while $P_{v,1}$
works in the opposite way when $y$ is larger than $z$, as in the
left figure. These phenomena for VIX volatility surfaces may occur
because $Y_{t}$ quickly converges $Z_{t}$ in our model. If $y>z$,
the spot volatility might decline in the short term, which means that
the VIX SITM calls would probably not be exercised. If $y<z$, the
spot volatility might increase for a short while, which would increase
the value of the VIX SOTM call. This is simply the mechanism of our
model showing diverse expression for the short-term VIX market in
spite of the constraint \eqref{eq:relation_ours}.

\section{Concluding remarks}

In this paper, we study consistent and efficient pricing of SPX and
VIX options under a new two-factor stochastic volatility model. Specifically,
this two-factor model is proposed by adding a fast mean-reverting
factor into the Heston model. Doing so facilitates joint pricing of
the SPX option and the VIX option. In practice, joint modeling of
both options is important, because an arbitrage relationship exists
between the SPX option market and the VIX option market. Moreover,
joint modeling leads to a calibration based on extensive market data,
including SPX data and VIX data. Since our analytic solutions are
derived as one-dimensional integrals, it is obvious that the pricing
solutions are computationally very efficient. Our experiment using
hundreds of thousands of options shows that the model reduces the
errors by 9.9-16.5\%, compared to the single-scale model of Heston.
The error reduction is possible because the additional factor reflects
short-term impacts on the market, which is difficult to achieve with
only one factor.

In fact, non-affine models have been less studied for explicit pricing
formulas because the involved problems are hard to solve. But the
models are more suitable to express actual dynamics (see \citet{mencia2013valuation}
and \citet{kaeck2013continuous}). In this context, our models should
be extended to have a non-affine form. We leave the topic as future
work. 

\section*{Acknowledgment}

Jaegi Jeon received financial support from the National Research Foundation
of Korea (NRF) of the Korean government (Grant No. NRF-2019R1I1A1A01062911).
Geonwoo Kim received financial support from the NRF (Grant No. NRF-2017R1E1A1A03070886).
Jeonggyu Huh received financial support from the NRF (Grant No. NRF-2019R1F1A1058352). 

\section*{Declarations of interest : none.}

\section*{References}

\bibliographystyle{elsarticle-harv}
\bibliography{ref}

\appendix

\section{\label{sec:appendix_a}Derivation of analytic formula for SPX options}

Putting $P_{s}^{\epsilon}=P_{s,0}+\sqrt{\epsilon}P_{s,1}+\epsilon P_{s,2}+\epsilon\sqrt{\epsilon}P_{s,3}+\cdots$
into the PDE (\ref{eq:spx_full_PDE}), 
\begin{align*}
 & \frac{1}{\epsilon}\mathcal{L}_{s,0}P_{s,0}+\frac{1}{\sqrt{\epsilon}}\left(\mathcal{L}_{s,0}P_{s,1}+\mathcal{L}_{s,1}P_{s,0}\right)+\left(\mathcal{L}_{s,0}P_{s,2}+\mathcal{L}_{s,1}P_{s,1}+\mathcal{L}_{s,2}P_{s,0}\right)\\
 & \qquad\qquad\qquad\qquad\qquad\qquad+\sqrt{\epsilon}\left(\mathcal{L}_{s,0}P_{s,3}+\mathcal{L}_{s,1}P_{s,2}+\mathcal{L}_{s,2}P_{s,1}\right)+\cdots=0.
\end{align*}
As mentioned in Theorem \ref{spx_theorem}, we intend for the sum
of the leading term $P_{s,0}$ and the first non-zero correction term
$P_{s,1}$ to approximate $P_{s}^{\epsilon}$ within accuracy $O\left(\epsilon\right)$,
that is, $\left|P_{s}^{\epsilon}-\left(P_{s,0}+\sqrt{\epsilon}P_{s,1}\right)\right|<C\epsilon$
for some positive $C$. For this purpose, the following equations
should be satisfied: 
\begin{gather}
\mathcal{L}_{s,0}P_{s,0}=0,\label{eq:0.5_order}\\
\mathcal{L}_{s,0}P_{s,1}+\mathcal{L}_{s,1}P_{s,0}=0,\label{eq:1_order}\\
\mathcal{L}_{s,0}P_{s,2}+\mathcal{L}_{s,1}P_{s,1}+\mathcal{L}_{s,2}P_{s,0}=0,\label{eq:1.5_order}\\
\mathcal{L}_{s,0}P_{s,3}+\mathcal{L}_{s,1}P_{s,2}+\mathcal{L}_{s,2}P_{s,1}=0.\label{eq:2_order}
\end{gather}
$P_{s,0}$ should not depend on $y$ so that the Poisson equation
(\ref{eq:0.5_order}) has a solution. This means that $P_{s,1}$ also
does not depend on $y$ owing to (\ref{eq:1_order}). From (\ref{eq:1.5_order}),
we then obtain 
\begin{equation}
\mathcal{L}_{s,0}P_{s,2}+\mathcal{L}_{s,2}P_{s,0}=0.\label{eq:1.5_order_2}
\end{equation}
Thus, the centering condition for $\mathcal{L}_{s,0}$ for the Poisson
equation (\ref{eq:1.5_order_2}) results in the following PDE (\ref{eq:PDE_p_s_0})
for $P_{s,0}$:

\begin{equation}
\left\langle \mathcal{L}_{s,2}P_{s,0}\right\rangle =\left\langle \mathcal{L}_{s,2}\right\rangle P_{s,0}=0,\label{eq:p_0_PDE_app}
\end{equation}
\[
P_{s,0}\left(T,x,z\right)=\left(x-K\right)^{+},
\]
where $\left\langle \cdot\right\rangle $ is the expectation with
respect to the invariant distribution $\Phi$ for $Y_{t}$, that is,
$\left\langle f\right\rangle =\int f\left(y\right)\Phi\left(y\right)dy$.
Moreover, because $Y_{t}$ is a CIR process, $Y_{\infty}\sim\Gamma\left(\gamma,\nu^{2}\right)$,
which means that the density function for $Y_{\infty}$ is 
\[
\Phi\left(w\right)=\frac{1}{\nu^{2\gamma}\Gamma\left(\gamma\right)}w^{\gamma-1}e^{-\left(w/\nu^{2}\right)}\boldsymbol{1}_{\left(0,\infty\right)}\left(w\right),
\]
where $\gamma=\frac{z}{\nu^{2}}$. The centering condition indicates
that $\left\langle g\right\rangle =0$ should hold for Poisson equation
$\mathcal{L}f+g=0$. The condition is necessary for a Poisson equation
to have a solution.

In addition, (\ref{eq:1.5_order_2}) yields the following equation:
\begin{align}
P_{s,2} & =-\mathcal{L}_{s,0}^{-1}\left(\mathcal{L}_{s,2}-\left\langle \mathcal{L}_{s,2}\right\rangle \right)P_{s,0}+c\left(t,x,z\right)\nonumber \\
 & =-\frac{1}{2}x^{2}\phi\left(y,z\right)\partial_{xx}^{2}P_{s,0}+c\left(t,x,z\right),\label{eq:P_s_2}
\end{align}
where 
\begin{equation}
\mathcal{L}_{s,0}\phi\left(y,z\right)=y-z.\label{eq:poisson_eq}
\end{equation}
In fact, we show $\phi\left(y,z\right)=z-y$, the proof of which is
given in \ref{sec:appendix_b}. On one hand, if we put (\ref{eq:P_s_2})
into (\ref{eq:2_order}) and use the centering condition for $\mathcal{L}_{s,0}$
in (\ref{eq:2_order}), we achieve the PDE (\ref{eq:PDE_p_s_1}) for
$P_{s,1}$ in the following ways: 
\begin{align}
\left\langle \mathcal{L}_{s,2}\right\rangle P_{s,1} & =-\left\langle \mathcal{L}_{s,1}P_{s,2}\right\rangle \nonumber \\
 & =\frac{1}{\sqrt{2}}\eta\nu\left\langle y\phi_{y}\left(y,z\right)\right\rangle x\partial_{x}\left(x^{2}\partial_{xx}^{2}P_{0}\right)\nonumber \\
 & =W_{3}zx\partial_{x}\left(x^{2}\partial_{xx}^{2}\right)P_{s,0}.\label{eq:p_1_PDE_app}
\end{align}
with the corresponding final condition 
\[
P_{s,1}\left(T,x,z\right)=0,
\]
where $W_{3}=-\frac{1}{\sqrt{2}}\eta\nu$ (recall that $W_{3}^{\epsilon}=-\frac{1}{\sqrt{2}}\eta\nu\sqrt{\epsilon}$,
i.e., $W_{3}=W_{3}^{\epsilon}/\sqrt{\epsilon}$).

On the other hand, if $\xi:=2z$, equations (\ref{eq:p_0_PDE_app})
and (\ref{eq:p_1_PDE_app}) are transformed as follows, and they are
associated with the PDE operator $\mathcal{L}_{H}$ for the Heston
model, whose parameters are the mean reversion rate of $\kappa$,
the long-run variance $2\theta$, the volatility of variance $\sqrt{2}\sigma$,
and the correlation $\sqrt{2}\rho$ between stock price and its variance.
\begin{gather}
\mathcal{L}_{H}\tilde{P}_{s,0}\left(t,x,\xi\right)=0,\nonumber \\
\mathcal{L}_{H}\tilde{P}_{s,1}\left(t,x,\xi\right)=\frac{1}{2}W_{3}\xi x\partial_{x}\left(x^{2}\partial_{xx}^{2}\right)\tilde{P}_{s,0}\left(t,x,\xi\right),\label{eq:p_1_PDE_xi}
\end{gather}
where 
\[
\mathcal{L}_{H}=\partial_{t}+rx\partial_{x}+\kappa\left(2\theta-\xi\right)\partial_{\xi}+\frac{1}{2}x^{2}\xi\partial_{xx}^{2}+\frac{1}{2}\left(\sqrt{2}\sigma\right)^{2}\xi\partial_{\xi\xi}^{2}+\left(\frac{1}{\sqrt{2}}\rho\right)\left(\sqrt{2}\sigma\right)x\xi\partial_{x\xi}^{2}-r.
\]
Similar to \citet{fouque2011fast}, by utilizing the feasibility of
the Heston model, we can achieve the following solutions of the abovementioned
PDEs: 
\begin{gather*}
\tilde{P}_{s,0}\left(t,x,\xi\right)=\frac{e^{-r\tau}}{2\pi}\int_{\mathbb{}}e^{-ikq}\hat{G}\left(\tau,k,2\xi\right)\hat{h}\left(k\right)dk,\\
\tilde{P}_{s,1}\left(t,x,\xi\right)=\frac{e^{-r\tau}}{2\pi}\int_{\mathbb{}}e^{-ikq}b\left(k\right)\left(\kappa\theta\hat{f}_{0}\left(\tau,k\right)+2z\hat{f}_{1}\left(\tau,k\right)\right)\hat{G}\left(\tau,k,\xi\right)\hat{h}\left(k\right)dk,
\end{gather*}
where $\tau=T-t$, $q=r\tau+\log x$, 
\begin{align*}
b\left(k\right) & =-\frac{1}{2}W_{3}\left(ik^{3}+k^{2}\right),\\
\hat{f}_{0}\left(\tau,k\right) & =\int_{0}^{\tau}\hat{f}_{1}\left(t,k\right)dt,\\
\hat{f}_{1}\left(\tau,k\right) & =\int_{0}^{\tau}\left(\frac{g\left(k\right)e^{sd\left(k\right)}-1}{g\left(k\right)e^{\tau d\left(k\right)}-1}\right)^{2}e^{d\left(k\right)\left(\tau-s\right)}ds,
\end{align*}
and $\hat{G}\left(\tau,k,z\right)$, $\hat{h}\left(k\right)$, $d\left(k\right)$,
and $g\left(k\right)$ are defined in Theorem \ref{spx_theorem}.
To compute $\tilde{P}_{s,0}$ and $\tilde{P}_{s,1}$, numerical integrations
need to be associated with a single integration and a triple integration.
However, as in the foregoing discussion in \ref{subsec:An-approximate-analytic},
numerical methods for the triple integration are too computationally
intensive. Fortunately, we can further simplify $\hat{f}_{0}$ and
$\hat{f}_{1}$, because the right-hand side of (\ref{eq:p_1_PDE_xi})
is linear with respect to $\xi$. The induction process is given in
detail as follows: 
\begin{align*}
\hat{f}_{1}\left(\tau,k\right) & =\int_{0}^{\tau}\frac{\left(g\left(k\right)e^{sd\left(k\right)}-1\right)^{2}}{\left(g\left(k\right)e^{\tau d\left(k\right)}-1\right)^{2}}e^{d\left(k\right)\left(\tau-s\right)}ds\\
 & =\frac{e^{\tau d\left(k\right)}}{\left(g\left(k\right)e^{\tau d\left(k\right)}-1\right)^{2}}\int_{0}^{\tau}\frac{\left(g\left(k\right)e^{sd\left(k\right)}-1\right)^{2}}{e^{sd\left(k\right)}}ds\\
 & =\left(\frac{e^{\tau d\left(k\right)}}{\left(g\left(k\right)e^{\tau d\left(k\right)}-1\right)^{2}}\right)\left(\frac{1}{d\left(k\right)}\left(g\left(k\right)^{2}e^{\tau d\left(k\right)}-2\tau d\left(k\right)g\left(k\right)-e^{-\tau d\left(k\right)}+1-g\left(k\right)^{2}\right)\right)+1\\
 & =\frac{e^{\tau d\left(k\right)}\left(g\left(k\right)^{2}\left(e^{\tau d\left(k\right)}-1\right)-2\tau d\left(k\right)g\left(k\right)+1\right)-1}{d\left(k\right)\left(g\left(k\right)e^{\tau d\left(k\right)}-1\right)^{2}}.
\end{align*}
If $I\left(\tau,k\right):=\frac{e^{\tau d\left(k\right)}}{\left(g\left(k\right)e^{\tau d\left(k\right)}-1\right)^{2}}$,
$\hat{f}_{1}\left(\tau,k\right)=I\left(t,k\right)\int_{0}^{t}\frac{1}{I\left(s,k\right)}ds$
from the second line of the above equations. Thus, we can obtain

\begin{align*}
\hat{f}_{0}\left(\tau,k\right) & =\int_{0}^{\tau}I\left(t,k\right)\int_{0}^{t}\frac{1}{I\left(s,k\right)}dsdt\\
 & =\int_{0}^{\tau}\frac{1}{I\left(s,k\right)}\int_{s}^{\tau}I\left(t,k\right)dtds\\
 & =\int_{0}^{\tau}\frac{\left(g\left(k\right)e^{sd\left(k\right)}-1\right)^{2}}{e^{sd\left(k\right)}}\left(\frac{1}{d\left(k\right)g\left(k\right)-d\left(k\right)g\left(k\right)^{2}e^{\tau d\left(k\right)}}-\frac{1}{d\left(k\right)g\left(k\right)-d\left(k\right)g\left(k\right)^{2}e^{sd\left(k\right)}}\right)ds\\
 & =\frac{2\tau d\left(k\right)g\left(k\right)+g^{2}\left(k\right)-1}{d^{2}\left(k\right)g\left(k\right)\left(g\left(k\right)e^{\tau d\left(k\right)}-1\right)}+\frac{\tau d\left(k\right)g\left(k\right)-g\left(k\right)-1}{d^{2}\left(k\right)g\left(k\right)}.
\end{align*}

\section{\label{sec:appendix_b}Solution for Poisson equation (\ref{eq:poisson_eq})}

By the spectral theory, solution $\phi$ for Poisson equation (\ref{eq:poisson_eq})
can be found as follows: 
\[
\phi\left(y,z\right)=z-y.
\]
Now, we briefly explain the way to find $\phi$. It is known that
the operator $\mathcal{L}_{s,0}$ has the eigenvalues $\lambda_{n}=-n$
$\left(n\in\mathbb{N}\right)$ and the family of eigenfunctions $\psi_{n}$
associated with eigenvalue $\lambda_{n}$ (i.e., $\mathcal{L}_{s,0}\psi_{n}=\lambda_{n}\psi_{n}$)
is 
\begin{align*}
\psi_{n}\left(y\right) & =\sqrt{\frac{n!\Gamma\left(\gamma\right)}{\Gamma\left(n+\gamma\right)}}L_{n}\left(\frac{y}{\nu^{2}}\right),
\end{align*}
where $\gamma=\frac{z}{\nu^{2}}$ and $L_{n}$ is an $n$-order Legendre
polynomial, that is, $L_{n}\left(w\right)=\frac{w^{1-\gamma}e^{w}}{n!}\frac{d^{n}}{dw^{n}}\left(w^{n+\gamma-1}e^{-w}\right)$.
It is also known that $\psi_{n}$ form a complete orthogonal basis
of the Hilbert space $L^{2}\left(\Phi\right)$. $\Phi$ is the invariant
distribution of $Y_{t}$, which is defined in the foregoing section.
Thus, $y-z\in L^{2}\left(\Phi\right)$ can be expressed as 
\begin{equation}
y-z=\sum_{n\ge0}c_{n}\psi_{n}\left(y\right),\label{eq:eigen_expansion}
\end{equation}
where $c_{n}=\left\langle \left(y-z\right)\psi_{n}\right\rangle $.
For any $n$, $c_{n}$ are explicitly calculated as follows: 
\begin{equation}
c_{n}=\begin{cases}
-\nu\sqrt{z} & {\rm if}\;n=1\\
0 & {\rm otherwise}
\end{cases}\label{eq:fourier_coefficient}
\end{equation}
We provide the induction process for the above formula \eqref{eq:fourier_coefficient}.
First, $c_{0}$ can be easily obtained as follows: 
\[
c_{0}=\left\langle \left(y-z\right)\psi_{0}\right\rangle =\left\langle y-z\right\rangle =0.
\]
In addition, $c_{1}$ is computed as follows: 
\begin{align*}
c_{1} & =\left\langle \left(y-z\right)\psi_{1}\right\rangle \\
 & =\sqrt{\frac{\Gamma\left(\gamma\right)}{\Gamma\left(1+\gamma\right)}}\int_{0}^{\infty}\left(y-z\right)\left(\gamma-\frac{y}{\nu^{2}}\right)\Phi\left(y\right)dy\\
 & =\sqrt{\frac{1}{\gamma}}\int_{0}^{\infty}\left(-\frac{1}{\nu^{2}}y^{2}+\left(\frac{z}{\nu^{2}}+\gamma\right)y-\gamma z\right)\Phi\left(y\right)dy\\
 & =-\nu\sqrt{z}.
\end{align*}
It is also proved that $c_{n}$ for $n\geq2$ are zero, because we
can show the following equation: 
\begin{align*}
 & \int_{0}^{\infty}y\left[\left(\frac{y}{\nu^{2}}\right)^{1-\gamma}e^{\frac{y}{\nu^{2}}}\frac{d^{n}}{dy^{n}}\left(\left(\frac{y}{\nu^{2}}\right)^{n+\gamma-1}e^{-\frac{y}{\nu^{2}}}\right)\right]\left[y^{\gamma-1}e^{-\left(y/\nu^{2}\right)}\right]dy\\
 & =\nu^{2\gamma-2}\int_{0}^{\infty}y\frac{d^{n}}{dy^{n}}\left(\left(\frac{y}{\nu^{2}}\right)^{n+\gamma-1}e^{-\frac{y}{\nu^{2}}}\right)dy\\
 & =-\nu^{2\gamma-2}\int_{0}^{\infty}\frac{d^{n-1}}{dy^{n-1}}\left(\left(\frac{y}{\nu^{2}}\right)^{n+\gamma-1}e^{-\frac{y}{\nu^{2}}}\right)dy\\
 & =-\nu^{2\gamma-2}\left[\frac{d^{n-2}}{dy^{n-2}}\left(\left(\frac{y}{\nu^{2}}\right)^{n+\gamma-1}e^{-\frac{y}{\nu^{2}}}\right)\right]_{0}^{\infty}\\
 & =\nu^{2\gamma-2}\left[\left(n+\gamma-1\right){\rm th\ degree\ polynomial\ without\ constant\ term}\times e^{-\frac{y}{\nu^{2}}}\right]_{0}^{\infty}\\
 & =0.
\end{align*}
Therefore, we eventually obtain the following $\phi$ by (\ref{eq:eigen_expansion})
and (\ref{eq:fourier_coefficient}): 
\begin{align*}
\phi\left(y,z\right) & =\frac{c_{1}\psi_{1}\left(y\right)}{\lambda_{1}}\\
 & =\frac{z}{\gamma}\left(\gamma-\frac{y}{\nu^{2}}\right)\\
 & =z-y.
\end{align*}

\begin{sidewaystable}

\section{\label{sec:errors_of_test}Errors in the empirical test}

\subfloat[In-sample errors for 2016--2017 options data]{\centering{}%
\begin{tabular}{ccccccccccccccccc}
\toprule 
 &  & \multicolumn{3}{c}{$\tau<0.05$} & \multicolumn{3}{c}{$0.05\leq\tau<0.1$} & \multicolumn{3}{c}{$0.1\leq\tau<0.2$} & \multicolumn{3}{c}{$\tau\geq0.2$} & \multicolumn{3}{c}{total}\tabularnewline
 &  &  &  &  &  &  &  &  &  &  &  &  &  &  &  & \tabularnewline
\midrule 
 &  & Heston & ours & o/h & Heston & ours & o/h & Heston & ours & o/h & Heston & ours & o/h & Heston & ours & o/h\tabularnewline
SPX & mean & 0.454 & 0.478 & 105\% & 0.509 & 0.419 & 82.2\% & 0.522 & 0.466 & 89.3\% & 0.494 & 0.557 & 113\% & 0.489 & 0.441 & 90.1\%\tabularnewline
 & std & 0.490 & 0.424 &  & 0.636 & 0.443 &  & 0.652 & 0.552 &  & 0.551 & 0.714 &  & 0.587 & 0.437 & \tabularnewline
 &  &  &  &  &  &  &  &  &  &  &  &  &  &  &  & \tabularnewline
\midrule
 &  & Heston & ours & o/h & Heston & ours & o/h & Heston & ours & o/h & Heston & ours & o/h & Heston & ours & o/h\tabularnewline
VIX & mean & 0.278 & 0.266 & 95.6\% & 0.402 & 0.344 & 85.5\% & 0.423 & 0.372 & 88.1\% & 0.437 & 0.391 & 89.5\% & 0.380 & 0.330 & 86.8\%\tabularnewline
 & std & 0.242 & 0.249 &  & 0.260 & 0.250 &  & 0.257 & 0.249 &  & 0.254 & 0.249 &  & 0.261 & 0.251 & \tabularnewline
 &  &  &  &  &  &  &  &  &  &  &  &  &  &  &  & \tabularnewline
\bottomrule
\end{tabular}}\\
 \subfloat[Out-of-sample errors for 2018 options data]{\centering{}%
\begin{tabular}{ccccccccccccccccc}
\toprule 
 &  & \multicolumn{3}{c}{$\tau<0.05$} & \multicolumn{3}{c}{$0.05\leq\tau<0.1$} & \multicolumn{3}{c}{$0.1\leq\tau<0.2$} & \multicolumn{3}{c}{$\tau\geq0.2$} & \multicolumn{3}{c}{total}\tabularnewline
 &  &  &  &  &  &  &  &  &  &  &  &  &  &  &  & \tabularnewline
\midrule 
 &  & Heston & ours & o/h & Heston & ours & o/h & Heston & ours & o/h & Heston & ours & o/h & Heston & ours & o/h\tabularnewline
SPX & mean & 0.500 & 0.472 & 94.5\% & 0.535 & 0.441 & 82.4\% & 0.526 & 0.452 & 85.8\% & 0.485 & 0.470 & 97.0\% & 0.521 & 0.453 & 87.0\%\tabularnewline
 & std & 0.538 & 0.373 &  & 0.813 & 0.369 &  & 0.871 & 0.412 &  & 0.720 & 0.475 &  & 0.717 & 0.371 & \tabularnewline
 &  &  &  &  &  &  &  &  &  &  &  &  &  &  &  & \tabularnewline
\midrule
 &  & Heston & ours & o/h & Heston & ours & o/h & Heston & ours & o/h & Heston & ours & o/h & Heston & ours & o/h\tabularnewline
VIX & mean & 0.321 & 0.259 & 80.6\% & 0.380 & 0.320 & 84.1\% & 0.393 & 0.338 & 86.1\% & 0.402 & 0.358 & 89.0\% & 0.369 & 0.308 & 83.5\%\tabularnewline
 & std & 0.294 & 0.247 &  & 0.301 & 0.232 &  & 0.282 & 0.235 &  & 0.267 & 0.240 &  & 0.300 & 0.236 & \tabularnewline
 &  &  &  &  &  &  &  &  &  &  &  &  &  &  &  & \tabularnewline
\bottomrule
\end{tabular}}

\caption{Errors in empirical test based on 2016--2018 options data. The results
are shown separately by time to maturity. The abbreviation ``o/h''
means the value obtained by dividing the error for our model by the
corresponding error for the Heston model.}
\end{sidewaystable}

\end{document}